\documentclass[12pt]{amsart}
\usepackage{verbatim}
\usepackage{enumerate}  

\usepackage{mathtools}
\usepackage{todonotes}
\usepackage{hyperref}

\numberwithin{equation}{section}

\newcommand{\ndash}{\nobreakdash-\hspace{0pt}}

\newcommand{\ii}{{\mathrm{i}}}
\newcommand{\dd}{{\mathrm{d}}}

\newcommand{\EE}{\mathrm{e}}

\newtheorem{Thm}{Theorem}[section]
\newtheorem{Prop}[Thm]{Proposition}
\newtheorem{Lem}[Thm]{Lemma}

\newtheorem*{Thm*}{Theorem}
\newtheorem*{Lem*}{Lemma}

\theoremstyle{remark}
\newtheorem{Rem}[Thm]{Remark}
\newtheorem*{Ack}{Acknowledgment}

\theoremstyle{definition}

\newtheorem{Def}[Thm]{Definition}

\newcommand{\bbR}{{\mathbb{R}}}

\newcommand{\de}{\partial}

\newcommand{\bv}{{\boldsymbol{v}}}

\newcommand{\balpha}{\boldsymbol{\alpha}}
\newcommand{\be}{{\boldsymbol{e}}}

\newcommand{\bsigma}{\boldsymbol{\sigma}}

\newcommand{\bomega}{{\boldsymbol{\omega}}}

\newcommand{\calL}{\mathcal{L}}

\newcommand{\calV}{\mathcal{V}}

\def\gpd{\,\lower1pt\hbox{$\longrightarrow$}\hskip-.24in\raise2pt
               \hbox{$\longrightarrow$}\,}

\let\Hat=\widehat

\usepackage[bibstyle=alphabetic,citestyle=alphabetic,useprefix,giveninits=true, sorting=ynt, sortcites, minbibnames=99,maxbibnames=99,backend=biber]{biblatex}  
\renewbibmacro{in:}{} 
\bibliography{bibliography}  
\begin{document}
\title{4D Palatini--Cartan Gravity in Hamiltonian Form}

\begin{abstract}
In this note the Hamiltonian formulation of four-dimensional gravity, in the Palatini--Cartan formalism, is recovered by elimination of an auxiliary field appearing as part of the connection.
\end{abstract}

\author{Alberto S. Cattaneo}
\address{Institut f\"ur Mathematik, Universit\"at Z\"urich\\
Winterthurerstrasse 190, CH-8057 Z\"urich, Switzerland}  
\email{cattaneo@math.uzh.ch}
\author{Giovanni Canepa}
\address{INFN, sezione di Firenze, via Sansone 1, 50019 Sesto Fiorentino (FI), Italy}
\email{giovanni.canepa.math@gmail.com}

\thanks{ASC acknowledges partial support of the SNF Grant No.\ 200020\_192080 and of the Simons Collaboration on Global Categorical Symmetries. This research was (partly) supported by the NCCR SwissMAP, funded by the Swiss National Science Foundation, and is based upon work from COST Action 21109 CaLISTA, supported by COST (European Cooperation in Science and Technology) (\href{www.cost.eu}{www.cost.eu}), MSCA-2021-SE-01-101086123 CaLIGOLA, and MSCA-DN CaLiForNIA - 101119552. GC acknowledges partial support of the SNF Grants No.\ P5R5PT\_222221 and No.\ P500PT\_203085.}

\maketitle

\setcounter{tocdepth}{2}
\tableofcontents

\allowdisplaybreaks

\section{Introduction}
In the Hamiltonian formalism, Palatini--Cartan (PC) gravity in four dimensions consists of a theory with vanishing evolution Hamiltonian and several constraints---generating internal gauge transformations and diffeomorphisms, both tangent and transversal to the Cauchy surface---see \cite{PhysRevD.58.124029,Alexandrov_2000,BarrosESa2001,Date:2008rb,Rezende:2009sv,Bodendorfer_2013, CS2019, CCS2020, MERC:2019, MRC:2019, ACOTZ00} and references therein. The goal of this note is to recover this not from a study of the constraints that arise in the boundary analysis, but via the elimination of an auxiliary field that happens to be present in PC gravity (with some assumptions essentially requiring that we are working on a globally hyperbolic space--time).

To start with, we recall the
PC action
\[
S[\be,\bomega]=\int_M \left(\frac{1}{2} \be^{2} F_{\bomega}  + \frac\Lambda{4!} \be^4\right),
\]
where $\be$ is a coframe (a.k.a.\ a tetrad or a vierbein) and $\bomega$ is a metric connection, which is regarded as an independent field. The parameter $\Lambda$ is the cosmological constant.

We focus on a space--time manifold $M$ of the form $\Sigma\times I$, where $\Sigma$ is a $3$\ndash manifold, possibly with boundary, and $I$ is an interval. Under the assumption that the metric induced on $\Sigma\times\{t\}$ by the coframe is positive\ndash definite $\forall t\in I$,\footnote{Our result actually extends to the weaker assumption that the induced metrics are nondegenerate.} we show that a part of the connection plays the role of an auxiliary field that is set to zero by its own Euler--Lagrangian (EL) equations. As a consequence, we conclude that PC theory is classically equivalent to a Hamiltonian system for fields on $\Sigma$ (the restriction of $\be$ and part of the restriction of $\bomega$) with zero Hamiltonian and several first\ndash class constraints (the torsion\ndash free condition and the constraints known as the momentum and Hamiltonian constraint).

If $\Sigma$ has a boundary, then there is a de facto Hamiltonian, which corresponds to the ADM mass.

All this fits nicely with the study of the constraints of PC theory (see \cite{CS2019,CCS2020} and references therein),  as well as with the analogue study in Einstein--Hilbert gravity (see \cite[Chapter VIII]{MR3379262} and references therein).

In the companion paper \cite{CC25a}, we extend this study to the BV setting and also to higher dimensions, proving that the BV PC theory is strongly equivalent (via a BV pushforward \cite{CMR2015}) to the AKSZ \cite{AKSZ} theory obtained in \cite{CCS2020b} from the BFV description of the reduced phase space of PC gravity.

\begin{Ack}
    We thank Simone Speziale for very useful comments.
\end{Ack}

\section{Classical Palatini--Cartan formalism in four dimensions}
In the Palatini--Cartan (PC) formulation of gravity, the fundamental ingredients are a coframe (a.k.a.\ a tetrad or a vierbein) and a connection, instead of a metric as in the Einstein--Hilbert formulation. The first idea (Cartan) is to use  bases, declared to be orthonormal, instead of metrics. A basis contains however more information than a metric because we can act on it by orthogonal transformations leaving the metric unchanged. This introduces a gauge freedom, so the PC theory turns out to be a gauge theory, and one has to introduce a metric connection,\footnote{By this we mean a connection for the orthogonal frame bundle $PM$. Using the coframe, we can then associate to it an affine connection, which is then automatically compatible with the induced metric.} which is treated as independent field (as in Palatini's formulation of general relativity). It is only via some equations of motion that the connection is forced to correspond to the Levi-Civita connection for the induced metric on space--time.

We denote the coframe by $\be$ and the connection by $\bomega$. The PC action (in four dimensions) reads
\begin{equation}\label{e:S}
S[\be,\bomega]=\int_M \left(\frac{1}{2} \be^{2} F_{\bomega}  + \frac\Lambda{4!} \be^4\right),
\end{equation}
where $\Lambda$ (the cosmological constant) is a fixed real number. Let us explain our notations, both in the intrinsic form, which we will use throughout, and in their local coordinate expression. We start with the latter.

\subsection{Local coordinates}
In local coordinates, the coframe $\be$ has components $e_\mu^a$ with $\mu$ and $a$ ranging from $1$ to $4$. The matrix $(e_\mu^a)$ is required to be invertible.
The Greek index is understood as a space--time index, whereas the Latin index as an internal one. The components of the connection $\bomega$ are $\omega_\mu^{ab}$ with the condition $\omega_\mu^{ba}=-\omega_\mu^{ab}$. 
The curvature $F_\bomega$ of $\bomega$ has components $F_{\mu\nu}^{ab}$ and is antisymmetric both in the lower and in the upper indices. It is defined as
\[
F_{\mu\nu}^{ab} = \de_{[\mu}\omega_{\nu]}^{ab} + 
\omega_{[\mu}^{ac}\omega_{\nu]}^{db}\eta_{cd},
\]
where $\eta$ is the Minkowski metric. 
Upon choosing an orientation for $M$, the action $\eqref{e:S}$ then reads
\[
S[\be,\bomega]=\int_M \epsilon^{\mu\nu\rho\sigma}\epsilon_{abcd}
\left(\frac{1}{2} e_\mu^ae_\nu^b F_{\rho\sigma}^{cd}  + \frac\Lambda{4!} e_\mu^ae_\nu^b e_\rho^ce_\sigma^d\right)
\;d^4x.
\]

\subsection{Intrinsic formulation}
We start choosing a vector bundle $\calV$ over $M$ isomorphic to the tangent bundle $TM$ ($\calV$ is often called the ``fake tangent bundle'') endowed with a fiberwise Minkowski metric $\eta$.

Equivalently, we may choose a reference Lorentzian metric on $M$ and consider the related orthonormal frame bundle $PM$. Then we define $\calV$ as the associated bundle $PM\times_G (V,\eta)$ with $(V,\eta)$ the Minkowski space and $G=SO(3,1)$.

A coframe is then an isomorphism $\be\colon TM\to\calV$. We denote by $\bomega$ a connection form for the orthonormal frame bundle $PM$ and by $F_\bomega$ its curvature. 
A useful notation we use throughout is 
\[
\Omega^{k,l}(M)\coloneqq \Gamma(M,\Lambda^k T^*M\otimes\Lambda^l\calV)
=\Omega^k(M)\otimes_{C^\infty(M)}\Gamma(M,\Lambda^l\calV),
\]
where $\Gamma$ denotes the space of sections, $\Lambda$ the exterior power,  and $\Omega(M)$ the space of differential forms. We call an element of $\Omega^{k,l}(M)$ a $(k,l)$\ndash form.

With this notation we see that $\be$ is a nondegenerate element of $\Omega^{1,1}(M)$. As the Lie algebra of orthogonal transformations may be identified with $\Lambda^2V$ (viz., the skew-symmetric matrices), the space of connections may be viewed as modeled on $\Omega^{1,2}(M)$, upon choice of a fixed reference connection $\bomega_0$. Finally, the curvature $F_\bomega$ may be viewed as an element of $\Omega^{2,2}(M)$. In summary,
\begin{align*}
\be &\in \Omega_\text{nd}^{1,1}(M),\\
\bomega &\in \bomega_0 + \Omega^{1,2}(M),\\
\intertext{and}
F_\bomega &\in \Omega^{2,2}(M),
\end{align*}
where the index nd stands for nondegenerate.

The argument of the integral in \eqref{e:S} is the $(4,4)$\ndash form obtained by taking the (nondisplayed) wedge products both in $\Lambda^\bullet T^*M$ and $\Lambda^\bullet\calV$.\footnote{Our convention for signs in the tensor product of these two algebras is by the total degree (Deligne convention). Namely, if $\alpha$ is a $(k,l)$\ndash form and $\beta$ an $(r,s)$\ndash form, we have $\alpha\beta=(-1)^{(k+l)(r+s)}\beta\alpha$.} One can show (see, e.g., 
\cite[Section 2.1]{CCS2020}) that the space
$\Omega^{4,4}(M)$ of $(4,4)$\ndash forms is canonically identified with the space of densities on $M$, so the integrand
\[
\calL = \frac{1}{2} \be^{2} F_{\bomega}  + \frac\Lambda{4!} \be^4
\]
is indeed a Lagrangian density. Note that the intrinsic formulation does not require $M$ to be orientable.

If $M$ is orientable, we may choose an orientation and orient $\calV$ via its defining isomorphism with $TM$. This yields an isomorphism of $\Lambda^4\calV$ with the trivial line bundle and of $\Omega^{4,4}(M)$ with the space of top forms on $M$. The Lagrangian density is then reinterpreted as a Lagrangian top form, and the action $S$ is defined integrating it over $M$ with the given orientation.

\subsection{Relation with the Einstein--Hilbert formalism}
A coframe $\be$ defines a metric $g$. One intrinsic way to see this is that $\be$ is a bundle map $TM\to\calV$. We then define $g$ as the pullback by $\be$ of the internal metric $\eta$: $g=\be^*\eta$. Equivalently, we have $g=(\be,\be)$ where $(\ ,\ )$ denotes the inner product defined by $\eta$. In local coordinates, we have $g_{\mu\nu}=e_\mu^ae_\nu^b\eta_{ab}$.\footnote{\label{f:G0}If we take care of physical dimensions, things are a bit more delicate. The fields $\bomega$ and $\be$ are dimensionsless. In fact $[\bomega,\bomega]$ must have the dimension of $\dd\bomega$, which implies that $\bomega$ is dimensionsless, since the exterior derivative $\dd=\dd x^\mu\frac\de{\de x^\mu}$ is so. On the other hand, the Lagrangian density must be dimensionsless, which implies that $\be$ is also so. But this means that the components $e_\mu$ in the expansion $\be=e_\mu\dd x^\mu$ have the dimension of inverse length, if we give the local coordinates $x^\mu$ the dimension of a length. On the other hand, the components $g_{\mu\nu}$ of the metric are dimensionsless, as implied by the formula $\dd s^2 = g_{\mu\nu}\dd x^\mu\dd x^\nu$. Therefore, the correct relation between the coframe and the metric is  $g_{\mu\nu}=G_0 e_\mu^ae_\nu^b\eta_{ab}$, where the constant $G_0$ has the dimension of a length squared. Up to a factor, $G_0$ can be identified with the gravitational constant.} The nondegeneracy of $\be$ is equivalent to the nondegeneracy of $g$. Also note that $g$ has the same signature as $\eta$. This shows that the space of metrics (of the given signature) on $M$ is isomorphic to the space of coframes $\Omega_\text{nd}^{1,1}(M)$ modulo orthogonal gauge transformations. Moreover, we may pullback $\bomega$ by $\be$ obtaining a metric connection $\nabla$ for $(M,g)$.

We now have to check the dynamical equivalence of the two theories. The Euler--Lagrange equations for $S$ read
\begin{align*}
\dd_\bomega(\be^2) &= 0,\\
\be F_\bomega + \frac\Lambda{3!}\be^3&=0,
\end{align*}
where $\dd_\bomega$ denotes the covariant derivative.\footnote{In local coordinates, 
$(\dd_\bomega\be)_{\mu\nu}^a=\de_{[\mu}^{\phantom{a}}e_{\nu]}^a + \omega_{[\mu}^{ab}e_{\nu]}^c\eta_{bc}$.}
By the Leibniz rule, the first equation is equivalent to $\be\dd_\bomega\be=0$, which turns out to be equivalent to $\dd_\bomega\be=0$ thanks to the nondegeneracy condition on $\be$ \cite{CCS2020}

Given $\be$, there is a unique $\bomega_\be$ solving the eqution $\dd_\bomega\be=0$. This corresponds to the induced metric connection $\nabla_g$ being torsion free, i.e., the Levi-Civita connection. If we now insert the solution $\bomega_e$ in the second Euler--Lagrange equation, and we rewrite it in terms of $g$, we get the Einstein equation (without matter and with cosmological constant $\Lambda$).

\section{PC gravity on a cylinder}\label{s:cyl}
We now specialize to the case when space--time $M$ is of the form $\Sigma\times I$, where $I=[t_0,t_1]$ is an interval and $\Sigma$ is a $3$\ndash manifold. We will interpret $I$ as a time interval, even though the results of this paper hold in more generality; in other words, we are thinking of $M$ as of a globally hyperbolic manifold. In this section, $\Sigma$ is assumed to be closed (i.e., compact without boundary). 

\newcommand\bbe{\underline\be}
\newcommand\bbomh{{\Hat{\underline\bomega}}}
\newcommand\bbomega{{\underline\bomega}}
\newcommand\ddd{\underline\dd}
\newcommand\bbv{\underline\bv}
\newcommand\bbsigma{\underline\bsigma}
\newcommand\uF{\underline F}
We denote by $t$ the coordinate on $I$ and by $\partial_\text{n}$ the vector field $\frac\de{\de t}$. (We use the index n to stress that this is the direction ``normal'' to $\Sigma$.) We call a $(k,l)$\ndash form $\underline\balpha$ horizontal if it satisfies $\iota_{\partial_\text{n}}\underline\balpha=0$, where $\iota$ denotes contraction with a vector field. We will use underlines to denote horizontal forms.

A $(k,l)$\ndash form $\balpha$ on $\Sigma\times I$ can be written uniquely as $\balpha=\alpha_\text{n}\dd t + \underline\balpha$ with $\underline\balpha$ horizontal.
\footnote{In local coordinates, $\balpha=\alpha_{\mu_1\dots\mu_k}\dd x^{\mu_1}\cdots\dd x^{\mu_k}$. If we denote by $i$ the indices from $1$ to $3$ labeling coordinates on $\Sigma$ and use the index $4$ for the coordinate on $I$, we have $\underline\balpha=\alpha_{i_1\dots i_k}\dd x^{i_1}\cdots\dd x^{i_k}$ and $\alpha_\text{n}=\alpha_{i_1\dots i_{k-1}4}\dd x^{i_1}\cdots\dd x^{i_{k-1}}$.} 
Note that $\alpha_\text{n}$ is a $(k-1,l)$\ndash form, whereas $\underline\balpha$ is still a $(k,l)$\ndash form. 
For the fields $\be$ and $\bomega$, we denote the decomposition by\footnote{The notation $\bbomh$ with a hat over $\bbomega$ is for later convenience.}
\begin{subequations}\label{e:beebbomega}
\begin{align}
\be &= e_\text{n}\dd t + \bbe,\\
\bomega &= \omega_\text{n}\dd t + \bbomh.
\end{align}
\end{subequations}
We may then rewrite the action as
\begin{equation}\label{e:Sintermdiate}
S = \int_M \left(
e_\mathrm{n} \left(\bbe\, \uF_{\bbomh} + \frac\Lambda{3!}\bbe^3\right)
+\frac12\bbe^2\left(-\de_\mathrm{n}\bbomh+\ddd_\bbomh\omega_\mathrm{n}\right)
\right)\dd t,
\end{equation}
where we have used the decomposition
\[
F_\bomega= \uF_{\bbomh} + \left(-\de_\mathrm{n}\bbomh+\ddd_\bbomh\omega_\mathrm{n}\right)
\dd t.
\]
Here $\ddd_\bbomh$ denotes the horizontal covariant derivative with respect to the connection $\bbomh$, which only has horizontal components; similarly, we denote by by
$\ddd$ the horizontal differential.\footnote{\label{fn:locdiff}In local coordinates, 
this simply means we only use Latin (space) indices. For example,
$(\ddd_\bbomh\omega_\mathrm{n})_{i}^{ab}=\de_i\omega_\mathrm{n}^{ab}+\omega_i^{ac} \omega_\mathrm{n}^{db}\eta_{cd}-\omega_i^{bc} \omega_\mathrm{n}^{da}\eta_{cd}$.} 

\begin{Rem}[The degeneracy problem]
The term $\frac12\bbe^2\de_\mathrm{n}\bbomh$ is the analogue of the $p\dot q$ term in the Hamiltonian version of an action functional. It has however a problem: it is invariant under certain changes in $\bbomh$, which leads to a degeneracy (this is intimately related to the degeneracy of the boundary structure observed in \cite{CS2019}). Namely, if we vary $\bbomh$ by a horizontal $(1,2)$\ndash form $\bbv$, we get a term $\frac12\bbe^2\de_\mathrm{n}\bbv$, which can be rewritten as $\frac12\bbe\de_\mathrm{n}(\bbe\,\bbv)-\frac12\bbe\,\bbv\de_\mathrm{n}\bbe$. If $\bbv$ satisfies $\bbe\,\bbv=0$, which actually admits nontrivial solutions, then $\frac12\bbe^2\de_\mathrm{n}\bbv$ vanishes. Below we show how to deal with this problem, essentially following \cite{CCS2020} (see also \cite[Section 7]{CPencyPC} for a shorter review).
\end{Rem}

First we pick a $(0,1)$\ndash form $\epsilon_\text{n}$ normalized by $(\epsilon_\text{n},\epsilon_\text{n})=-1$. We can view $\epsilon_\text{n}$ as a reference time-like basis vector for the internal Minkowski space $(V,\eta)$. {}From now on we require the field $\bbe$ to satisfy the additional nondegeneracy condition
\begin{equation}\label{e:e3ennot0}
\bbe^3\epsilon_\text{n}\not=0
\end{equation}
everywhere. This corresponds to saying that at each point  the three components of $\bbe$ together with $\epsilon_\text{n}$ form a basis of the fiber of $\calV$. Therefore, we have a uniquely determined, nowhere vanishing function $\mu$ and a uniquely determined  horizontal vector field $Z$ (i.e., $\iota_Z\dd t =0$) such that
\begin{equation}\label{e:en}
e_\text{n} = \mu\epsilon_\text{n}+\iota_Z\bbe.
\end{equation}

To proceed, we now have to make the crucial assumption introduced in \cite{CCS2020}. We denote by $g^\Sigma$, at each value of $t\in I$, the 
restriction to $\Sigma\times\{t\}$ of the metric $g$ defined by $\be$ (equivalently, $g^\Sigma$ is the 
tensor defined 
by pulling back the internal metric $\eta$ via $\bbe$).\footnote{In local coordinates, $g^\Sigma_{ij} =\underline e_i^a\underline e_j^b\eta_{ab}$.}

\begin{Def}\label{d:metnondeg}
We say that $\bbe$ is metric nondegenerate if it satisfies \eqref{e:e3ennot0} and if
$g^\Sigma$ is positive\ndash definite, hence a Riemannian metric, for every $t\in I$.
\end{Def}
\begin{Rem} 
The condition amounts to saying that $M=\Sigma\times I$ is globally hyperbolic. For all the considerations in this paper, we could actually use the weaker condition that $g^\Sigma$ be nondegenerate, but not necessarily positive definite.\footnote{One mild advantage is that the condition that $g^\Sigma$ be positive-definite together with the time-like normalization $(\epsilon_n,\epsilon_n)=-1$ implies the nondegeneracy condition \eqref{e:e3ennot0}, which then need not be imposed separately.}
\end{Rem}

Under the assumption that $\bbe$ is metric nondegenerate, it is proven in \cite[Section 4.1]{CCS2020} that there is a unique decomposition\footnote{For a related analysis in the Dirac formalism, see \cite{BarrosESa2001,Alexandrov_2000}.}
\begin{equation}\label{e:bom}
\bbomh = \bbomega + \bbv
\end{equation}
by a uniquely determined horizontal connection $\bbomega$ and a uniquely determined horizontal $(1,2)$\ndash form $\bbv$ satisfying
\begin{subequations}\label{e:structev}
\begin{align}
\epsilon_\text{n}\,\ddd_\bbomega\bbe &= \bbe\,\bbsigma,\label{e:struct}\\
\bbe\,\bbv &=0,\label{e:ev}
\end{align}
\end{subequations}
where $\ddd_\bbomega$ denotes again the horizontal covariant derivative,\footnote{In local coordinates, as in footnote~\ref {fn:locdiff},
$(\ddd_\bbomega\bbe)_{ij}^a=\de_{[i}^{\phantom{a}}\underline e_{j]}^a + \underline\omega_{[i}^{ab}\underline e_{j]}^c\eta_{bc}$.} and $\bbsigma$ is a horizontal $(1,1)$\ndash form. Equation \eqref{e:struct} is called the structural constraint in \cite{CCS2020}.
Finally we introduce the $(0,2)$\ndash form
\begin{equation}\label{e:w}
w \coloneqq -\omega_\text{n} - \iota_Z\bbv.
\end{equation}
We can summarize the first part of this section as the description of a change of variables from $(\be,\bomega)$, with $\be$ metric nondegenerate, to
$(\mu,Z,w,\bbe,\bbomega,\bbv)$.
\begin{Prop}\label{t:Sauxham}
We have
\[
S[\be,\bomega] = S_\mathrm{Ham}[\mu,Z,w,\bbe,\bbomega] + S_\mathrm{aux}[\mu,Z,\bbe,\bbv]
\]
with
\[
S_\mathrm{Ham}= \int_M \left(
-\frac12 \bbe^2\de_\mathrm{n}\bbomega
+ (\mu\epsilon_\mathrm{n}+\iota_Z\bbe) \left(\bbe\, \uF_{\bbomega} + \frac\Lambda{3!}\bbe^3\right)
+ w \bbe\,\ddd_\bbomega\bbe
\right)\dd t
\]
and
\[
S_\mathrm{aux}= \int_M \frac12 (\mu\epsilon_\mathrm{n}+\iota_Z\bbe)\bbe[\bbv,\bbv].
\]
\end{Prop}
In the above expressions we have used the following notations: $\de_\mathrm{n}$ is a shorthand notation for $\frac\de{\de t}$, $\uF_{\bbomega}$ is the horizontal part of the curvature of $\bbomega$, and $[\ ,\ ]$ denotes the Lie bracket in the orthogonal Lie algebra.
\begin{proof}
We start from \eqref{e:Sintermdiate}.
By \eqref{e:bom} we have
\[
\uF_{\bbomh} = \uF_{\bbomega} + \ddd_\bbomega\bbv + \frac12 [\bbv,\bbv],
\]
so
\[
S = \int_M \left(
e_\mathrm{n} \left(\bbe\, \uF_{\bbomega} + \bbe\,\ddd_\bbomega\bbv + \bbe\frac12 [\bbv,\bbv]+ \frac\Lambda{3!}\bbe^3\right)
+\frac12\bbe^2\left(-\de_\mathrm{n}\bbomh+\ddd_\bbomh\omega_\mathrm{n}\right)
\right)\dd t.
\]

We will now proceed to rewrite the single terms using \eqref{e:en}, \eqref{e:bom}, \eqref{e:structev}, and \eqref{e:w}.
{}From \eqref{e:ev} we get $\bbe\,\ddd_\bbomega\bbv=-\ddd_\bbomega\bbe\,\bbv$, so 
\[
e_\mathrm{n}\bbe\,\ddd_\bbomega\bbv = -(\mu\epsilon_\mathrm{n}+\iota_Z\bbe)\ddd_\bbomega\bbe\,\bbv=
-\mu\bbe\,\bbsigma\,\bbv-\iota_Z\bbe\,\ddd_\bbomega\bbe\,\bbv = -\iota_Z\bbe\,\ddd_\bbomega\bbe\,\bbv,
\]
where we have used again \eqref{e:struct} and \eqref{e:ev}.
{}From \eqref{e:ev} we get $\iota_Z\bbe\,\bbv=-\bbe\iota_Z\bbv$, so we have
\[
e_\mathrm{n}\bbe\,\ddd_\bbomega\bbv =-\iota_Z\bbv\,\bbe\,\ddd_\bbomega\bbe.
\]
The next term we consider is
\[
\bbe^2 \de_\mathrm{n}\bbomh= \bbe^2 \de_\mathrm{n}\bbomega + \bbe^2 \de_\mathrm{n}\bbv.
\]
{}From \eqref{e:ev} we get $\bbe\de_\mathrm{n}\bbv=-\de_\mathrm{n}\bbe\,\bbv$, so the last terms becomes
$-\be\de_\mathrm{n}\bbe\,\bbv$, which vanishes by \eqref{e:ev}.

Finally, since $\Sigma$ is closed, we have
\begin{equation}\label{e:bypartsSigmaclosed}
0 = \int_\Sigma \ddd(\bbe^2\omega_\mathrm{n}) = 
\int_\Sigma \ddd_\bbomh\bbe^2\,\omega_\mathrm{n}
+ \int_\Sigma \bbe^2\,\ddd_\bbomh\omega_\mathrm{n}.
\end{equation}
However, we have
\[
\frac12 \ddd_\bbomh\bbe^2=\bbe\,\ddd_\bbomh\bbe=
\bbe\,\ddd_\bbomega\bbe
+\bbe\,\bbv\cdot\bbe.
\]
A simple computation, see \cite{CCS2020}, shows that \eqref{e:ev} implies $\bbe\,\bbv\cdot\bbe=0$. Therefore, upon integrating by parts, we may replace $\frac12\bbe^2\ddd_\bbomh\omega_\mathrm{n}$ with 
$-\omega_\mathrm{n}\bbe\,\ddd_\bbomega\bbe$.

Collecting all the terms and expanding $e_\mathrm{n}$ as in \eqref{e:en} complete the proof.
\end{proof}

\subsection{Getting rid of the auxiliary term}\label{s:aux}
The field $\bbv$ only appears in the functional $S_\mathrm{aux}$, in which it enters quadratically and with no derivatives. 
The Euler--Lagrange equation for $\bbv$ is therefore a homogeneous linear equation. We claim that the only solution is $\bbv=0$. This immediately follows from the
\begin{Lem}\label{l:auxnd}
If $\bbe$ is metric nondegenerate, then $S_\mathrm{aux}$, viewed as a quadratic form in $\bbv$, is nondegenerate.
\end{Lem}
\begin{proof}
It is convenient to rewrite $S_\mathrm{aux}$ as $\int_M \frac12 e_\mathrm{n}\bbe[\bbv,\bbv]$. Since no derivative of $\bbv$ is present, we can study it pointwise at each $x\in M$. In a neighborhood of $x$, we can choose coordinates such that the metric components $g_{\mu\nu}(x)$ are equal to $\eta_{\mu\nu}$. The metric nondegeneracy condition of Definition~\ref{d:metnondeg} implies in turn that, up to an internal gauge transformation, we have $e_\mu^a(x)=\delta_\mu^a$. Applying this internal gauge transformation also to $\bbv$ and observing that the auxiliary Lagrangian is gauge invariant, we are left with the Lagrangian function\footnote{The Lagrangian density $\calL\mathrm{aux}$ is $L_\mathrm{aux}\,d^4x$.}
\[
L_\mathrm{aux} = \frac12\epsilon_{rstl}\epsilon^{ijk} \delta_n^r\delta_i^s[\bbv,\bbv]_{jk}^{tl}
= \frac12\epsilon_{itl}\epsilon^{ijk}[\bbv,\bbv]_{jk}^{tl}=[\bbv,\bbv]_{jk}^{jk},
\]
where everything is computed at the point $x$.
Recall that
\[
[\bbv,\bbv]_{jk}^{tl}= v_{[j}^{ta} v_{k]}^{bl} \eta_{ab},
\]
so
\[
L_\text{aux}=\psi-\phi
\]
with
\[
\phi = v_j^{ja} v_k^{kb} \eta_{ab}\qquad\text{and}\qquad\psi  = v_k^{ja} v_j^{kb} \eta_{ab}.
\]
We set $y_j^a\coloneqq  v_j^{ja}$ (no sum over $j$). Then
\[
\phi = \|y_1+y_2+y_3\|^2= \|y_1\|^2 + \|y_2\|^2 + \|y_3\|^2 
+ 2((y_1,y_2)+(y_1,y_3)+(y_2,y_3)),
\]
where we have denoted by $(\ ,\ )$ the inner product defined by $\eta$ and have written $\|y_j\|^2\coloneqq(y_j,y_j)$.
Next we compute $\psi$. We split the sum over $j$ and $k$ into one with $j=k$ and one with $j\not= k$.
We then have $\psi=\psi_=+\psi_{\not=}$ with 
\[
\psi_= = \|y_1\|^2 + \|y_2\|^2 + \|y_3\|^2.
\]
To compute $\psi_{\not=}$ more easily, we use \eqref{e:ev}, which for $e_\mu^a=\delta_\mu^a$ yields, as computed in \cite[Appendix A]{CCS2020}, $v_i^{a4}=0$, for all $i,a=1,2,3$, as well as the relations
\begin{equation}\label{e:vrelations}
v_3^{32} = -v_1^{12},\qquad v_2^{23} = -v_1^{13},\qquad v_3^{31}=-v_2^{21}.
\end{equation}
We may then choose the following $6$ independent components of $\bbv$: $v_1^{12}$, $v_1^{13}$, $v_1^{23}$, $v_2^{12}$, $v_2^{13}$, $v_3^{12}$.

Simply using $v_i^{a4}=0$, the antisymmetry of $v$ in the upper indices, and the fact that $(\eta_{ij})_{i,j=1,2,3}$ is the Euclidean metric, we get
\[
\psi_{\not=} = v_1^{23}v_2^{13} + v_1^{32} v_3^{12} + v_2^{31}v_3^{21}
= v_1^{23}v_2^{13} - v_1^{23} v_3^{12} + v_2^{13}v_3^{12}.
\]
Finally, we have
\[
L_\text{aux} = xy-xz+yz
- 2((y_1,y_2)+(y_1,y_3)+(y_2,y_3)),
\]
where have set $x=v_1^{23}$, $y=v_2^{13}$, and $z=v_3^{12}$. 

Using again $v_i^{a4}=0$, we see that $(y_i,y_j)$ is equal to the Euclidean inner product $(u_i,u_j)$
 of the $3$\ndash vectors $u_i$ given by  $u_i^j=v_i^{ij}$ with $i,j=1,2,3$ and no sum over $i$. 
Explicitly,
\[
u_1 = (0,v_1^{12},v_1^{13}), \qquad u_2 = (-v_2^{12},0,v_2^{23}),\qquad u_3=(-v_3^{13},-v_3^{23},0).
\]
Using the relations \eqref{e:vrelations}, we can rewrite them as
\[
u_1 = (0,v_1^{12},v_1^{13}), \qquad u_2 = (-v_2^{12},0,-v_2^{13}),\qquad u_3=(v_2^{12},-v_1^{12},0).
\]
Therefore,
\[
(u_1,u_2)=-(v_1^{13})^2,\quad (u_1,u_3)=-(v_1^{12})^2,\quad (u_2,u_3)=-(v_2^{12})^2.
\]
Setting $\alpha=v_1^{13}$, $\beta=v_1^{12}$, and $\gamma=v_2^{12}$, we finally have
\begin{equation}\label{e:Laux}
L_\text{aux} = xy-xz+yz
- 2(\alpha^2+\beta^2+\gamma^2),
\end{equation}
which is a nondegenerate quadratic form in the six independent variables $x,y,z,\alpha,\beta,\gamma$.
\end{proof}

\subsection{The Hamiltonian action}
Thanks to Lemma~\ref{l:auxnd}, $\bbv$ can be viewed as an auxiliary field which is set to zero by its Euler--Lagrange equations. 

Therefore, assuming metric nondegeneracy, 4D gravity (Palatini--Cartan or, equivalently, Einstein--Hilbert) is classically equivalent to the theory described by the action $S_\text{Ham}$ on the fields $(\mu,Z,w,\bbe,\bbomega)$.

We now proceed to examine this theory. The only time derivative, $\de_\mathrm{n}$, appears in the term
$\frac12 \bbe^2\de_\mathrm{n}\bbomega$. This implies that
$(\bbe,\bbomega)$ are the only dynamical fields, whereas $(\mu,Z,w)$ are Lagrange multipliers. 

Moreover, the Hamiltonian analysis in \cite{CS2019,CCS2020} shows that $\frac12 \bbe^2\de_\mathrm{n}\bbomega$ is indeed the 
analogue of the $p\dot q$ term in Hamiltonian mechanics, in the sense that, at fixed $t\in I$,
\begin{equation}\label{e:alphaSigma}
\alpha^\Sigma \coloneqq \int_\Sigma \frac12 \bbe^2\delta \bbomega
\end{equation}
is a potential for a symplectic form $\omega^\Sigma = \delta\alpha^\Sigma$. (Here $\delta$ denotes the exterior detivative on the space of fields $(\bbe,\bbomega)$ on $\Sigma$.)

The Lagrange multipliers $(\mu,Z,w)$ impose the following constraints
\begin{align}
\bbe\,\ddd_\bbomega\bbe &= 0,\\
\epsilon_\mathrm{n}\left(\bbe\, \uF_{\bbomega} + \frac\Lambda{3!}\bbe^3\right) &= 0,\\
\bbe_i \left(\bbe\, \uF_{\bbomega} + \frac\Lambda{3!}\bbe^3\right) &= 0.
\end{align}
The first, together with the structural constraint \eqref{e:struct}, is the torsion-free condition. The second and the third are known has the Hamiltonian and the momentum constraints, respectively. It is shown in \cite{CCS2020} that these three constraints are first class. Moreover, the first generates gauge transformations, the second diffeomorphisms transversal to $\Sigma$, and the third diffeomorphism along $\Sigma$.

Finally, there is no further term, which means that the Hamiltonian of the system vanishes, so there is no time evolution. This is of course consistent with general relativity's invariance under diffeomorphisms. 

To have a more interesting situation, we have to allow $\Sigma$ to have a boundary (possibly at infinity), which we consider in Section~\ref{s:boundaries}.

\subsection{Initial and final conditions}\label{s:ifcond}
The $1$\ndash form $\alpha^\Sigma$ in \eqref{e:alphaSigma} can in particular be considered at $\Sigma\times\{t_0\}$ and $\Sigma\times\{t_1\}$ and corresponds to the boundary terms appearing when computing the variation of the Hamiltonian action $S_\mathrm{Ham}$.

 If we use this form, the only initial and final conditions we may choose correspond to fixing the value of $\bbomega$. This is  unsatisfactory both physically, for it would be more natural to fix the coframe (hence the metric), and mathematically, for $\bbomega$ must satisfy the structural constraint \eqref{e:struct}, which depends on $\bbe$. 
 
 In order to fix instead $\bbe$ at the boundary, we have to change the $1$\ndash form $\alpha^\Sigma$ by an exact term $\delta f^\Sigma$. The problem is that $\bbomega$ is not a well-defined $1$\ndash form, for it is a connection. We can however write $\bbomega= \bbomega_0 + \bbomega'$, which we regard as a change of variables $\bomega\mapsto\bomega'$, where $\bbomega_0$ is a reference connection. 
Since $\bbomega_0$ is fixed, we have $\delta\bbomega_0=0$ and therefore
$\alpha^\Sigma = \int_\Sigma \frac12 \bbe^2\delta \bbomega'$.
We can now add to $S_\mathrm{Ham}$ (at $t=t_0$ and $t=t_1$, with appropriate signs) 
the term\footnote{This is the general way of modifying a theory buy a boundary term in order to be able to pass from Neumann to Dirichlet boundary conditions. For PC, see \cite[equation (4.7)]{Corichi_2016} and references therein.}  
\[
f^\Sigma=\int_\Sigma{\frac12\bbe^2\bbomega'}.
\]
This modifies the action,  which we now denote by $S_\mathrm{Ham}'$, in such a way that, when taking variations, the boundary terms are presently given by the $1$\ndash form 
\[
\alpha^{'\Sigma} \coloneqq \int_\Sigma  \bbomega'\bbe\delta\bbe
\]
at $t=t_0$ and at $t=t_1$, with appropriate signs. If $\bbe$ is fixed at the initial and final times, then the variation of $S_\mathrm{Ham}'$ has no boundary terms.\footnote{In principle, it is possible to add the term $f^\Sigma$ at only one of the two boundary components. In this case, at that boundary component we will fix $\bbe$, while at the other we will have to fix $\bbomega'$.}

\begin{Rem}\label{r:o00}
If we pick $\bbomega_0$ satisfying $\de_\mathrm{n}\bbomega_0=0$, then the modified Hamiltonian action  becomes, after integrating by parts the term $\frac12 \bbe^2\de_\mathrm{n}\bbomega$,
\begin{multline*}
S_\mathrm{Ham}'[\mu,Z,w,\bbe,\bbomega']\\= \int_M \left(
\bbomega'
 \bbe\de_\mathrm{n}\bbe
+ (\mu\epsilon_\mathrm{n}+\iota_Z\bbe) \left(\bbe\, \uF_{\bbomega_0 + \bbomega'} + \frac\Lambda{3!}\bbe^3\right)
+ w \bbe\,\ddd_{\bbomega_0 + \bbomega'}\bbe
\right)\dd t,
\end{multline*}
\end{Rem}

\section{Boundaries}\label{s:boundaries}
We now allow $\Sigma$ to have boundary (this makes $M=\Sigma\times[t_0,t_1]$ into a manifold with corners).
The only change in the above is that in the proof of Proposition~\ref{t:Sauxham}, equation \eqref{e:bypartsSigmaclosed}
has to be replaced by
\[
\int_{\de\Sigma}\bbe^2\omega_\mathrm{n}
 = \int_\Sigma \ddd(\bbe^2\omega_\mathrm{n}) = 
\int_\Sigma \ddd_\bbomh\bbe^2\,\omega_\mathrm{n}
+ \int_\Sigma \bbe^2\,\ddd_\bbomh\omega_\mathrm{n}.
\]
The result is that $S_\mathrm{Ham}$ 
gets the boundary correction
\[
S_\mathrm{b}=\int_{\de\Sigma\times[t_i,t_f]}\frac12\bbe^2\omega_\mathrm{n}\,\dd t.\footnote{It is worth noting that, as a function of $\be$, the boundary term $S_\mathrm{b}$ is the classical part of one of the generators of the corner $P_\infty$\ndash algebra 
discussed in \cite[Section 6]{CC24}, where $\omega_\mathrm{n}$ plays the role of a parameter.}
\]
As in Section~\ref{s:ifcond}, we should however be careful with boundary conditions, which, as things stand, can only fix $\omega_n$. 

If we are interested in the more physical boundary conditions that fix $\be$ on the boundary, we have to modify the original action $S$ of equation \eqref{e:S} to $S'=S+f^\de$ with
\[
f^\de=-\int_{\de M} \frac12 \be^2(\bomega-\bbomega_0),
\]
where $\bbomega_0$ is a reference connection as in Section~\ref{s:ifcond}. Note that $f^\de$ already includes the terms $f^\Sigma$ discussed there corresponding to the boundary portion $\Sigma\times\de[t_i,t_f]$. The term corresponding to the rest of the boundary is 
\[
-\int_{\de \Sigma\times[t_i,t_f]} \frac12 \be^2(\bomega-\bbomega_0) 
=-\int_{\de \Sigma\times[t_i,t_f]} \frac12 \bbe^2\omega_\mathrm{n}\,\dd t
+ \int_{\de \Sigma\times[t_i,t_f]}e_n\bbe\,\bbomega'\,\dd t.
\]

Note that $S_\mathrm{aux}$ gets no boundary terms, so the discussion in Section~\ref{s:aux} is unchanged, and we still find that the auxiliary field $\bbv$ is set to zero by its own EL equations.
Therefore, we get that the action $S'$ is equivalent to
\begin{multline*}
S_\mathrm{Ham}'[\mu,Z,w,\bbe,\bbomega']\\= \int_M \left(
\bbomega'
 \bbe\de_\mathrm{n}\bbe
+ (\mu\epsilon_\mathrm{n}+\iota_Z\bbe) \left(\bbe\, \uF_{\bbomega_0 + \bbomega'} + \frac\Lambda{3!}\bbe^3\right)
+ w \bbe\,\ddd_{\bbomega_0 + \bbomega'}\bbe
\right)\dd t\\
+ \int_{\de \Sigma\times[t_i,t_f]}e_n\bbe\,\bbomega'\,\dd t,
\end{multline*}
where, as in Remark~\ref{r:o00}, we have assumed $\de_\mathrm{n}\bbomega_0=0$ and have integrated by parts in the time direction.

Note that in the boundary term the fields $e_\mathrm{n}$ and $\bbe$ are fixed by the boundary conditions. When inserting into $S_\mathrm{Ham}'$ the solution to the Euler--Lagrange equations---so computing the Hamilton--Jacobi action---we may interpret the term
\[
\mathbf{M}\coloneqq
\int_{\de \Sigma}e_n\bbe\,\bbomega'_\mathrm{sol},
\]
as the energy (or mass) of the solution.

In particular, if we pick $e_\mathrm{n}=\epsilon_\mathrm{n}$ and $\bbe$ to correspond to Schwarzschild space, then $\mathbf{M}$ is the ADM \cite{ADM1959} mass.\footnote{For an overview of the ADM mass in the PC formalism, see \cite{Ashtekar2008,OS2020bis}.} Note that our discussion parallels that of \cite[Section 4.3]{P04} and references therein.

\subsection{The Schwarzschild mass}
We may test the formula for the ADM mass on the Schwarzschild solution. On the exterior of the black hole, space--time is $M=\Sigma\times\bbR$, with $\Sigma=(2Gm,+\infty)\times S^2$, with $G$ the gravitational constant and $m$ the black-hole mass. A choice of coframe corresponding to the Schwarzschild solution is
\[
\begin{split}
	\be^0 &=\frac1{\sqrt G_0} \sqrt{1-\frac{2Gm}{r}}\,\dd t,\\
	\be^1 &= \frac1{\sqrt G_0} \frac1{\sqrt{1-\frac{2Gm}{r}}}\,\dd r,\\
	\be^2 &= \frac1{\sqrt G_0}\,r\,\dd\theta, \\
	\be^3 &= \frac1{\sqrt G_0}\,r\sin\theta\,\dd\phi,
\end{split}
\]
with $r$ the coordinate on $(2Gm,+\infty)$, $\theta$ and $\phi$ the polar coordinates on $S^2$, and $G_0$ the constant introduced in footnote~\ref{f:G0}. The corresponding torsion-free connection has the following nonvanishing components
\[
\begin{split}
	\bomega^{12} &= \sqrt{1-\frac{2Gm}{r}} \,\dd\theta,\\
	\bomega^{13} &= \sqrt{1-\frac{2Gm}{r}}\, \sin\theta\,\dd\phi,\\
	\bomega^{23} &= \cos\theta\,\dd\phi.
\end{split}
\]
As reference connection $\bbomega_0$ we take the one corresponding to Minkowski space, i.e., the above for $m=0$. We then have
\begin{multline*}
	e_n\bbe(\bbomega-\bbomega_0)= e_n^0(\be^2(\bomega-\bbomega_0)^{31}+ \be^3(\bomega-\bbomega_0)^{12})\\=\frac{2r}{G_0}\sqrt{1-\frac{2Gm}{r}}\left(1-\sqrt{1-\frac{2Gm}{r}}
\right)\sin\theta\,\dd\theta\,\dd\phi.
\end{multline*}
Therefore,
\[
M(r)\coloneqq \int_{\{r\}\times S^2} e_n\bbe(\bbomega-\bomega_0) =
\frac{8\pi r}{G_0}\sqrt{1-\frac{2Gm}{r}}\left(1-\sqrt{1-\frac{2Gm}{r}}
\right).
\]
Finally, we get that the ADM mass corresponding to $\Sigma$ is
\[
\mathbf{M} = \lim_{r\to+\infty} M(r)- \lim_{r\to 2GM^+} M(r)=\frac{8\pi G}{G_0} m.
\]
This fixes $G_0=8\pi G$.

\section{Euclidean gravity}
All the results in this paper (i.e., all classical considerations)\footnote{See Section~\ref{s:qlor} for the problems at the quantum level.} have a straightforward generalization in the case of Euclidean gravity, where the metric has positive signature. Here are the minor differences:
\begin{enumerate}
\item The ``fake tangent bundle'' $\calV$ is now introduced as the associated bundle $PM\times_G (V,\eta)$ 
to the orthgonal frame bundle $PM$ for a reference Riemannian metric. Here $(V,\eta)$ is the Euclidean space and $G=SO(4)$.
\item The reference $(0,1)$\ndash form $\epsilon_\text{n}$ is normalized by $(\epsilon_\text{n},\epsilon_\text{n})=1$.
\item Metric nondegenracy as in Definition~\ref{d:metnondeg} is now automatical.
\end{enumerate}
On the other hand, all the computations proceed exactly as in the Lorentzian case. In particular, the auxiliary Lagrangian $L_\mathrm{aux}$ of \eqref{e:Laux} is exactly the same as in the Lorentzian case.

\section{Towards quantization}\label{s:qlor}
The first remark one can make is that in the functional integral for PC gravity we can first integrate out the auxiliary field $\bbv$.\footnote{\label{fn:qLor}Note that this integration has to be done for each fixed $(e_\mathrm{n},\bbe)$. It is a fiber integration, as $\bbv$ must satisfy the condition $\bbe\bbv=0$.} As the quadratic form $L_\mathrm{aux}$ in \eqref{e:Laux} is nondegenerate, the functional integral $\int \EE^{\frac\ii\hbar S_\mathrm{aux}}\,D\bbv$ is well\ndash defined (in the Lorentzian case).\footnote{In the Riemannian case, we should consider the functional integral $\int \EE^{-\frac1\hbar S_\mathrm{aux}}\,D\bbv$, which is not well\ndash defined, as the quadratic form in \eqref{e:Laux} is not positive\ndash definite.}

To be more precise, one first has to take care of the symmetries of the PC action. If we do this in the BV formalism \cite{BV1}, see \cite{CSPCH}, we can then extend the splitting discussed in Section~\ref{s:cyl} to all the BV fields. At this point we can view the integration of the auxiliary field $\bbv$ as a particular application of a BV pushforward \cite[Section 2]{CMR2015}. In the companion paper \cite{CC25a}, we perform these steps and show this way that the PC theory on a cylinder is equivalent to the AKSZ \cite{AKSZ}  theory associated to its boundary BFV data constructed in \cite[Section 4]{CCS2020b}. (In \cite{CC25a}, we will actually allow more general nondegeneracy conditions than in Definition~\ref{d:metnondeg} and will also consider the higher\ndash dimensional case.)

The second step towards quantization would be to compute the Hamilton--Jacobi action of theory presented in this paper, following \cite{CMW22}, and to use the corresponding BV--AKSZ formulation to start computing the perturbative functional integral, following \cite{CMW23}. The renormalization issues should start appearing at this point.

\begin{refcontext}[sorting=nyt]
    \printbibliography[] 
\end{refcontext}

\end{document}